\documentclass[11pt]{article}
\usepackage{amsmath,amsthm,amssymb,amsfonts}
\newcommand{\field}[1]{\mathbb{#1}}
\newcommand{\remove}[1]{}
\usepackage{longtable}
\usepackage{complexity}

\newtheorem{thm}{Theorem}[section]

\newtheorem{lem}[thm]{Lemma}
\newtheorem{define}[thm]{Definition}
\newtheorem{cor}[thm]{Corollary}
\newtheorem{obs}[thm]{Observation}

\newtheorem{prop}[thm]{Proposition}

\def\F{{\mathbb{F}}}

\def\N{{\mathbb{N}}}

\def\_{\,\,\,\,\,}

\def\poly{\textsf{poly}}

\def\sps{\Sigma\Pi\Sigma}
\def\spsp{\Sigma\Pi\Sigma\Pi}
\def\sp{\Sigma\Pi}
\newcommand{\hsp}{\Sigma\Pi\Sigma{\Pi}^{[a]}}

\newcommand{\hspn}{\Sigma\Pi\Sigma{\Pi}^{[\frac{t}{20}]}}
\newcommand{\dep}{{depth $4$}}
\def\psp{\Pi\Sigma\Pi}
\begin{document}


\mathchardef\mhyphen="2D 

\title{The Limits of Depth Reduction for Arithmetic Formulas: \\{It's all about the top fan-in}
}
\author{Mrinal Kumar\thanks{Department of Computer Science, Rutgers University.
Email: \texttt{mrinal.kumar@rutgers.edu}.}\and
Shubhangi Saraf\thanks{Department of Computer Science and Department of Mathematics, Rutgers University.
Email: \texttt{shubhangi.saraf@gmail.com}.}}

\date{}
\maketitle
\abstract{In recent years, a very exciting and promising method for proving lower bounds for arithmetic circuits has been proposed. This method combines the method of {\it depth reduction} developed in the works of Agrawal-Vinay~\cite{AV08}, Koiran~\cite{koiran} and Tavenas~\cite{Tavenas13},  and the use of the shifted partial derivative complexity measure developed in the works of Kayal~\cite{Kayal12} and Gupta et al~\cite{GKKS12}. 
These results inspired a flurry of other beautiful results and strong lower bounds for various classes of arithmetic circuits, in particular a recent work of Kayal et al~\cite{KSS13} showing superpolynomial lower bounds for {\it regular} arithmetic formulas via an {\it improved depth reduction} for these formulas. It was left as an intriguing question if these methods could prove superpolynomial lower bounds for general (homogeneous) arithmetic formulas, and if so this would indeed be a breakthrough in arithmetic circuit complexity. 

In this paper we study the power and limitations of depth reduction and shifted partial derivatives  for arithmetic formulas. We do it via studying the class of depth 4 homogeneous arithmetic circuits. We show: (1) the first {\it superpolynomial lower bounds} for the class of homogeneous depth 4 circuits with top fan-in $o(\log n)$. The core of our result is to show  {\it improved depth reduction} for these circuits. This class of circuits has received much attention for the problem of polynomial identity testing. We give the first nontrivial lower bounds for these circuits for any top fan-in $\geq 2$.  (2) We show that improved depth reduction {\it is not possible} when the top fan-in is $\Omega(\log n)$. In particular this shows that the depth reduction procedure of Koiran and Tavenas~\cite{koiran, Tavenas13} cannot be improved even for homogeneous formulas, thus strengthening the results of Fournier et al~\cite{FLMS13} who showed that depth reduction is tight for circuits, and answering some of the main open questions of ~\cite{KSS13, FLMS13}.   Our results in particular suggest that the method of improved depth reduction and shifted partial derivatives may not be powerful enough to prove superpolynomial lower bounds for (even homogeneous) arithmetic formulas. 

}

\thispagestyle{empty}

\newpage
\pagenumbering{arabic}



\section{Introduction}
 In a seminal paper in 1979, Valiant~\cite{Valiant79} laid out a neat theoretical framework for the study of resource bounded algebraic computation and defined the complexity classes $\VP$ and $\VNP$ as the algebraic analogs of $\P$ and $\NP$ respectively. Since then, the problem of understanding whether $\VNP$ is different from $\VP$ has been a problem of fundamental significance in Algebraic Complexity Theory. To show that $\VNP$ is different from $\VP$, it would suffice to show that the Permanent polynomial, which is a complete problem for $\VNP$~\cite{Valiant79} does not have polynomial sized arithmetic circuits. Unfortunately, not much progress has been made towards proving superpolynomial arithmetic circuit lower bounds for any explicit polynomial in spite of the intensive attention that the problem has received. In recent years much effort has been invested in proving lower bounds for restricted classes of arithmetic circuits. The hope is that understanding restricted classed might shed light on how to approach the much more general and seemingly harder problem. Small depth circuits are one such class which have been quite intensively studied from this perspective, and even for small depth circuits, we really only understand lower bounds for depth 2 circuits and some classes of depth 3 and depth 4 circuits~\cite{NW95, SW01, GrigorievKarpinski98, GKKS12, KSS13}. 

Recently a very promising and exciting new framework for proving lower bounds for arithmetic circuits
has emerged. The framework consists of two major components. Let $\mathcal C$ be the class of circuits one wants to prove lower bounds for. The first step is to show that any circuit in $\mathcal C$ can be efficiently {\it depth reduced} to a depth 4 circuit with bounded bottom fan-in ($\spsp^{[t]}$ circuit). This depth reduction procedure was introduced and developed in the works of Agrawal-Vinay~\cite{AV08}, Koiran~\cite{koiran} and Tavenas~\cite{Tavenas13}, building upon the initial depth reduction procedure of Valiant et al~\cite{VSBR83}. The second step is to prove strong lower bounds for $\spsp^{[t]}$ circuits using the {\it shifted partial derivative} complexity measure, which was developed in the works of Kayal~\cite{Kayal12} and Gupta et al~\cite{GKKS12}. Recently this framework was used successfully to prove the first superpolynomial lower bounds for {\it regular formulas}~\cite{KSS13}, and it seemed promising that such techniques could be used to prove lower bounds for more general classes such as {\it general} arithmetic formulas.  

In this paper, we successfully apply this framework to  prove the first superpolynomial lower bounds for homogeneous depth 4 circuits with bounded top fan-in. We prove our results via an {\it improved depth reduction}~\footnote{By depth reduction, we really mean a reduction to homogeneous depth 4 circuits with bounded bottom fan-in. So, it makes sense to talk of depth reduction for depth 4 circuits.}. We also show that if the bound on the top fan-in is relaxed (even by a small amount), then {\it efficient depth reduction is not possible}. In particular this suggests that the method of improved depth reduction + shifted partial derivatives seems to be not powerful enough to prove lower bounds for (even) homogeneous arithmetic formulas. This result strengthens the results in~\cite{KSS13, FLMS13}, and answers some of the main open questions posed in them. 

We now outline the major results and the sequence of events that build up to the results of this paper.  In the discussion in the rest of this section, we will refer to the class of circuits of depth 4 ($\spsp$ circuits) with bottom (product) fan-in bounded by $t$ as $\spsp^{[t]}$ circuits. 


\vspace{2mm} 
\noindent
{\bf Depth Reduction: } In a surprising result in 2008,  Agrawal and Vinay~\cite{AV08}  showed that any homogeneous polynomial which can be computed by a polynomial sized circuit of {\it arbitrary depth} can also be computed by subexponential sized homogeneous {\it depth $4$} $\spsp$ circuit. In other words, in order to prove superpolynomial (or even exponential) lower bounds for general arithmetic circuits, it suffices to prove exponential ($\exp(\Omega(n))$) lower bounds for just depth 4 arithmetic circuits\footnote{This result came as a big surprise, and indeed nothing like this is true in the Boolean world.}! 
In a follow up paper Koiran~\cite{koiran} improved the parameters of this depth reduction theorem and showed that in order to prove superpolynomial lower bounds for general arithmetic circuits, it suffices to prove a lower bound of the form $\exp(\omega(\sqrt{n}\log^2n))$ for homogeneous $\spsp^{[\sqrt n]}$ circuits (for polynomials of degree $n$). He also showed that to prove superpolynomial arithmetic formula lower bounds, it suffices to prove a slightly weaker lower bound of the form $\exp(\omega(\sqrt{n}\log n))$ for homogeneous $\spsp^{[\sqrt n]}$ circuits. Tavenas~\cite{Tavenas13} further refined the parameters of Koiran's result and showed that for circuits lower bounds also, a lower bound of $\exp(\omega(\sqrt{n}\log n))$ would suffice. This sequence of works laid out an approach towards  separating $\VNP$ from $\VP$ by just proving strong enough lower bounds for homogeneous circuits of {\dep}. In a recent work along this line, Gupta, Kamath, Kayal and Saptharishi~\cite{GKKS13} prove that strong enough lower bounds for {\it depth $3$} circuits suffice to show superpolynomial lower bounds for circuits of {\it arbitrary} depth, although in this case, we lose the property of homogeneity that was true for the reduction to {\dep}. This loss in homogeneity seems quite severe, at least with respect to proving lower bounds, and we know only weak lower bounds for non-homogeneous depth 3 circuits~\cite{SW01}. (For the rest of the paper, this depth reduction to non-homogeneous depth 3 circuits will not be relevant.)
More precisely, the results of Tavenas~\cite{Tavenas13} and Koiran~\cite{koiran} state the following. 

\begin{thm}[\cite{koiran, Tavenas13}]\label{thm:tavenas} 
Every polynomial size circuit of degree $n$ in $N$ variables can be transformed into an equivalent homogeneous $\spsp^{[t]}$ circuit with top fan-in\footnote{For $\spsp^{[t]}$ circuits where $t = \sqrt n$, observe that an upper bound of $\exp(O(\sqrt{n}\log N))$ on the top fan-in the circuit implies the same upper bound on size, since each product gate at the second layer computes a polynomial with at most $\exp(O(\sqrt{n}\log N))$ monomials. However for other values of $t$, the top fan-in bound will be the more relevant parameter for depth reduction.} at most $\exp(O(\frac{n}{t}\log N))$.
\end{thm}


\vspace{2mm}
\noindent
{\bf Depth $4$ Lower Bounds and $\VNP$ vs $\VP$: } In light of the results of Agrawal-Vinay~\cite{AV08}, Koiran~\cite{koiran} and Tavenas~\cite{Tavenas13}, proving lower bounds for homogeneous {\it depth $4$} circuits seems like an extremely promising direction to pursue in order to separate $\VNP$ from $\VP$. In a breakthrough result in this direction, Gupta, Kamath, Kayal and Saptharishi~\cite{GKKS12} proved that any homogeneous $\spsp^{[\sqrt n]}$ circuit computing the permanent must have size (and top fan-in) $\exp(\sqrt{n})$. This was strengthened in a more recent work of Kayal, Saha and Saptharishi~\cite{KSS13}, where it was shown that there is an explicit family of polynomials in $\VNP$ such that any homogeneous $\spsp^{[\sqrt n]}$ circuit computing it must have size (and top fan-in) at least $\exp(\Omega(\sqrt{n}\log n))$. More precisely,

\begin{thm}[\cite{GKKS12, KSS13}]\label{thm:gkks} 
For every $n$, there is an explicit family of polynomials in $\VNP$ in $N = \theta(n^2)$ variables and with degree $\theta(n)$ such that any homogeneous $\spsp^{[t]}$ circuit computing it must have top fan-in at least $\exp(\Omega(\frac{n}{t}\log N))$.
\end{thm}

The depth reduction results combined with the lower bounds for homogeneous $\spsp^{[t]}$ circuits is indeed a remarkable collection of results. As it stands, in order to separate $\VP$ from $\VNP$, any small asymptotic improvement in the exponent on either the lower bound front or on the depth reduction front would be sufficient! In fact for any class of circuits $\cal C$ for which we can improve the depth reduction parameters of Theorem~\ref{thm:tavenas}, we would get superpolynomial lower bounds for that class using Theorem~\ref{thm:gkks}.


Unfortunately, it seems that in general, we cannot hope for a better depth reduction. In a recent work, Fournier, Limaye, Malod and Srinivasan~\cite{FLMS13} gave an example of an explicit polynomial in $\VP$ (of degree $n$ and in $N= n^{O(1)}$ variables) such that any homogeneous $\spsp^{[t]}$ circuit computing it must have top fan-in at least $\exp(\Omega(\frac{n}{t}\log N))$. This immediately implies that the depth reduction parameters in the result of Tavenas~\cite{Tavenas13} are {\it tight} for circuits. This observation, along with the fact that the hard polynomial used by Kayal et al~\cite{KSS13} has a shifted partial derivative span only a polynomial factor away from the maximum possible value suggests that the technique of improving depth reduction and then using shifted partial derivatives may not be strong enough to separate $\VNP$ from $\VP$\footnote{The reason this statement is not completely formal is that we still do not know know if the upper bounds on the shifted partial derivative measure for $\spsp^{[t]}$ circuits is tight for all choices of derivatives and shifts, though the results of~\cite{FLMS13} and this paper show that they are indeed tight for many of the choices.}. In a recent result, Chillara and Mukhopadhyay~\cite{CM13} gave a clean unified way of way of lower bounding the shifted partial derivative complexities of the polynomials considered by~\cite{KSS13, FLMS13}.



 \vspace{2mm}
\noindent
 {\bf Formula Lower Bounds:} Even though improved depth reduction does not seem to be powerful enough to separate $\VNP$ from $\VP$, it is conceivable that it could lead to superpolynomial lower bounds for other interesting classes, for instance homogeneous arithmetic formulas, or even general arithmetic formulas. This hope was further strengthened when Kayal et al~\cite{KSS13} used these precise ideas to prove superpolynomial lower bounds for a restricted class of formulas which they called {\it regular} formulas. (Regular formulas are formulas which have alternating sum and product layers. Moreover, for every fixed layer, the fan-ins of the gates in that layer are the same and the formal degree of the formula is at most a constant times the formal degree of the polynomial being computed.) Kayal et al proved their result by showing that one can reduce any polynomial size regular formula to a $\spsp^{[t]}$ circuit (for a carefully chosen choice of $t$) of size asymptotically better in the exponent than the $\exp(\frac{n}{t}\log N)$ bound (which as we just discussed is known to be tight for circuits). This improvement in depth reduction immediately leads to superpolynomial lower bounds for regular formulas by using Theorem~\ref{thm:gkks}. 

Removing the restriction on regularity and proving superpolynomial lower bounds for general formulas or even general homogeneous formulas would be a huge step forward - it would be by far the strongest and most natural class of arithmetic circuits for which we would be able to prove lower bounds, and it would represent a real breakthrough. 
The authors of the two papers~\cite{KSS13, FLMS13} left as a tantalizing open question whether formulas (or even homogeneous formulas) can have better depth reduction than circuits (such as is true for regular formulas). If true, this would imply superpolynomial lower bounds for (homogeneous) formulas. Indeed it seemed quite likely to be true since at the face of it, regularity of formulas of formulas did not seem like such a severe restriction at all (and indeed this was argued to be the case). Perhaps it could be also be true that every formula could be reduced to a regular formula with only a polynomial blow up in size. If so, the improved depth reduction for formulas (and hence the lower bounds) would follow from the improved depth reduction of regular formulas. 

Thus to summarize, the main challenge that remained was to understand the limits of the techniques of depth reduction and shifted partial derivatives. In particular, are there any other interesting classes of circuits for which improved depth reduction is possible? Is improved depth reduction possible for arithmetic formulas?



\section{Our results}

In this paper we study the power and limitations of depth reduction for arithmetic formulas. We do this via  studying depth reduction for depth 4 arithmetic circuits\footnote{Since depth 4 arithmetic circuits are also equivalent to depth 4 arithmetic formulas upto a polynomial blow up in size, we will use the term circuits and formulas interchangeably when referring to depth 4 circuits.}. Let homogeneous $\spsp(r)$ circuits be the class of homogeneous depth 4 circuits with {\it top fan-in} bounded by $r$, and with {\it no restriction on the bottom fan-in}. This is a very natural class of circuits and is quite different in nature from $\spsp^{[t]}$ circuits. 

Our results are divided into two parts. In the first part we show the first superpolynomial lower bounds for homogeneous $\spsp(r)$ circuits when $r = o(\log n)$. The core of our result is an improved depth reduction result for these circuits. (As we pointed out, when we refer to `depth reduction', we really mean a reduction to homogeneous $\spsp^{[t]}$ circuits. Thus it makes sense to talk about a depth reduction for $\spsp$ circuits as well.) $\spsp(r)$ circuits have received significant attention for the problems of polynomial identity testing and polynomial reconstruction~\cite{KMSV10, SarafV11, GuptaKL12}, however prior to this work there were no nontrivial lower bounds for this class of circuits for any value of $r \geq 2$. 

In the second part we show that efficient depth reduction {\it is not possible} for homogeneous arithmetic formulas. We show this result by studying the very simple class of formulas given by homogeneous $\spsp(\log n)$ circuits.  We show that for this class of circuits, improved depth reduction is not possible. 
This shows that improved depth reduction is unfortunately not powerful enough to prove lower bounds for homogeneous $\spsp(\log n)$ circuits, and in particular not strong enough to prove lower bounds for homogeneous arithmetic formulas, answering the main open questions of~\cite{KSS13, FLMS13}. 

Informally, our main results are the following:
\\

\noindent{\bf Main Theorem 1 (Informal):}
{\it There is an explicit family of polynomials in $\VNP$ of degree $n$ in $N= n^{O(1)}$ variables such that for $r = o(\log n)$, any polynomial size homogeneous $\spsp(r)$ circuit computing it must have superpolynomial size.}
\\

At the core of the result is the following ``depth reduction" result:
\\

\noindent{\bf Improved Depth Reduction (Slightly wishful):}\footnote{Indeed the above statement is not quite true, and our reduction turns out to be much more subtle.  We do not depth reduce to a $\spsp^{[t]}$ circuit, but one in which the sum of degrees of any $\epsilon n/t$ product gates at the bottom is at most $\epsilon n$. This is a more refined notion and a slightly more general class of circuits than $\spsp^{[t]}$ circuits. We observe that the shifted partial derivative technique does not distinguish between these two kinds of circuits, and thus we are still able to obtain our lower bounds. Thus in spirit we still get depth reduction. In fact everywhere in this paper we could replace $\spsp^{[t]}$ circuits with this slightly more general class of circuits, and none of the results would be affected.}
{\it For $r = o(\log n)$, any polynomial size homogeneous $\spsp(r)$ circuit computing a polynomial of degree $n$ and in $N$ variables is equivalent to a homogeneous $\spsp^{[t]}$ circuit of size $\exp\left(o(\frac{n}{t}\log N)\right)$ for some choice of $t$ such that $\log^2 n \leq t \leq \epsilon n$. }
\\


Observe that the parameters of the depth reduction we obtain above improve upon the parameters of depth reduction given by~\cite{koiran, Tavenas13}.

We also show that when $r = \Omega(\log(n))$, depth reduction as above is no longer true. 
\\

\noindent{\bf Main Theorem 2 (Informal)}
{\it For $r = \Omega(\log n)$, there exists an explicit family of polynomials $\{{\cal Q}_n\}_n$ computed by a $\poly(n)$ size homogeneous $\spsp(r)$ circuit (and hence homogeneous formula) of degree $n$ and in $N$ variables, such that for every $t$ such that $\omega(\log n) \leq t \leq \epsilon n$, any homogeneous $\spsp^{[t]}$ circuit computing ${\cal Q}_{n}$ must have top fan-in at least $\exp\left(\Omega(\frac{n}{t}\log N)\right)$.
}
\\

An immediate consequence of this result is that the depth reduction procedure of Tavenas~\cite{Tavenas13} is tight for homogeneous arithmetic formulas (strengthening the results of~\cite{FLMS13}). 

At the core of our result is a {\it hierarchy} theorem for homogeneous $\spsp^{[t]}$ circuits which shows that homogeneous $\spsp^{[t]}$ circuits are a much richer class than homogeneous $\spsp^{[t/20]}$ circuits. We state this result more formally in Theorem~\ref{thm:hierarchy}. 

It was shown in~\cite{KSS13} that any ABP (even non homogeneous) can be converted to a regular formula with a quasipolynomial blow up in size. If one could improve this transformation even slightly for formulas or even for homogeneous formulas, this would imply superpolynomial lower bounds for formulas/homogeneous formulas. Another consequence of our results is that such an improvement is not possible. We build upon the results of~\cite{KSS13} and show that the conversion of general formulas to regular formulas must incur a quasipolynomial blow up in size. 
\\

\noindent{\bf Theorem (Conversion to Regular Formulas is Tight)}
{\it For $r = \Omega(\log n)$, there exists an explicit family of polynomials $\{{\cal Q}_n\}_n$ computed by a $\poly(n)$ size homogeneous $\spsp(r)$ circuit (and hence also homogeneous formula) of degree $n$ and in $N= n^{O(1)}$ variables, such that any regular formula computing ${\cal Q}_n$ must have size $N^{\Omega(\log n)}$.
}
\\



In the sections below we formally state our results and  elaborate on them in greater detail, as well as highlight some of the interesting corollaries of our proof techniques.

\subsection{Lower bounds for $\spsp(r)$ circuits, $r = o(\log n)$}

In the first part of the paper, we explore the limits of computation of depth 4 homogeneous circuits when the restriction for the bottom fan-in is removed. 
For the general model of (even homogeneous) $\spsp$ circuits, only extremely weak lower bounds seem to be known. Even PIT for $\spsp$ circuits is known only when the top fan-in is constant and the circuit is multilinear (in the multilinear case, the degree of the polynomials computed must anyway be bounded by the number of variables, and hence, multilinearity is a much bigger restriction than homogeneity\footnote{In all the results of this paper, the restriction of homogeneity can be replaced by the restriction that all gates in the circuit compute polynomials of degree at most $n$.}). The problem of showing lower bounds for depth 4 circuits with bounded top fan-in is hence a problem that is simpler than derandomizing PIT for the same model (at least in the black box model), and it seems to be the first crucial step in that direction. Moreover, even when the top fan-in is 2, prior to this work there were no lower bounds known. Unlike the class of depth $3$ circuits with bounded top fan-in which cannot even compute all polynomials irrespective of the size of the circuit, the class of $\spsp(r)$ circuits is complete (even for $r=1$). For more discussion on the completeness of this class, see Appendix~\ref{app:completeness}. 

We consider homogeneous $\Sigma\Pi\Sigma\Pi(r)$ circuits, which are depth 4 homogeneous circuits whose {\it top fan-in} is bounded by $r$. When $r$ is a constant we prove exponential lower bounds\footnote{In the rest of the paper, by exponential lower bound we will mean a lower bound of the form $2^{n^\epsilon}$ for some constant $\epsilon$.} for the class of  $\Sigma\Pi\Sigma\Pi(r)$ circuits, and for any $r = o(\log n)$ we show superpolynomial lower bounds for $\spsp(r)$ circuits\footnote {It is important to observe that the reduction of a polynomial sized homogeneous $\spsp$ circuit with arbitrary bottom fan-in to a homogeneous $\spsp$ circuit with bounded bottom fan-in as given by the results of~\cite{AV08,koiran} can lead to circuits of size $\exp(\Omega(n/t)\log n)$ and so Theorem~\ref{thm:gkks} does not imply any nontrivial lower bounds for it.}.
In particular, we prove the following theorem:

\begin{thm}{\label{thm:lb1}}
There exists an explicit family of polynomials in $\VNP$, $\{NW_n\}_n$, such that for each $n$, $NW_n$ has degree $\theta(n)$, and number of variables $\theta(n^2)$ and such that the following holds: Let $C$ be a homogeneous $\Sigma\Pi\Sigma\Pi(r)$ circuit that computes $NW_n$. Let $s$ be the size of $C$. Then $$s \geq \exp\left(n^{\Omega(1/r)}\log n\right).$$
\end{thm}

Prior to this result, we are not aware of any such lower bounds for depth 4 circuits even when the top fan-in $r$ equals $2$. 

\paragraph{Lower bounds for homogeneous $\Sigma\Pi\Sigma\Pi^\ast$ circuits: }

Another class of circuits we are able to prove a lower bound for is the class of depth 4 circuits where each product at the second layer (from the top)  has the same degree sequence of incoming polynomials, and there is no restriction on the top fan-in. 

For any degree sequence $\mathcal D = D_1, D_2, \ldots, D_k$ of non-negative integers such that $\sum D_i = n$, we study the class of homogeneous $\Sigma\Pi^{\mathcal D}\Sigma\Pi$ circuits, which are homogeneous circuits where each $\Pi$ gate at the second layer is restricted to having its inputs be polynomials whose sequence of degrees is precisely $\mathcal D$. 
We show that for {\it every} degree sequence $\mathcal D$, any $\Sigma\Pi^{\mathcal D}\Sigma\Pi$ circuit computing $NW_n$ (an explicit family of polynomials in $\VNP$) must have size at least $\exp\left(n^{\epsilon}\right)$, for some fixed absolute constant $\epsilon$ independent of $\mathcal D$. 
In particular, let the class of $\spsp^\ast$ circuits be the union of the classes of $\Sigma\Pi^{\mathcal D}\Sigma\Pi$ for all $\mathcal D$. Then our lower bounds hold for homogeneous $\spsp^\ast$ circuits as well. 

\begin{thm}{\label{thm:lb2}}
There exists an explicit family of polynomials in $\VNP$, $\{NW_n\}_n$, such that for each $n$, $NW_n$ has degree $\theta(n)$, and number of variables $\theta(n^2)$ and such that the following holds: Let $C$ be a homogeneous $\Sigma\Pi\Sigma\Pi^\ast$ circuit that computes $NW_n$.  Let $s$ be the size of $C$. 
Then $$s \geq \exp\left(n^{\epsilon}\right), $$ for some fixed absolute constant $\epsilon > 0$. 
\end{thm}


 \subsection{Depth reduction is tight for $\spsp(r)$ circuits, $r = \Omega (\log n)$}
The main question that was left open by both the works of~\cite{KSS13} and ~\cite{FLMS13} was to understand whether an improved depth reduction was possible for general (homogeneous) arithmetic formulas. 

In particular, the following tantalizing questions naturally emerge and were left as open questions by the works of~\cite{KSS13} and~\cite{FLMS13}.

 \begin{itemize}
\item Can the depth reduction by Koiran and Tavenas~\cite{koiran, Tavenas13} be improved for formulas: In other words, can one show that for every polynomial of degree $n$ and in $N = n^{O(1)}$ variables which has a polynomial sized (homogeneous) formula, it can be reduced to a $\spsp^{[\sqrt n]}$ circuit of size $N^{o(\sqrt{n})}$? 
\item Can every homogeneous arithmetic formula be converted to a regular formula with only a polynomial blow up in its size?
\end{itemize}

A positive answer to any of the above questions would suffice in proving superpolynomial lower bounds for general homogeneous arithmetic formulas. We settle both the questions and show that unfortunately neither is true. 

 We settle these questions by constructing a an explicit family of polynomials $\{{\cal Q}_n\}_n$, where ${\cal Q}_n$ is a polynomial in $\theta(n^2)$ variables and is of degree $\theta(n)$, such that for each $n$, ${\cal Q}_n$ can be computed by a {\it polynomial sized homogeneous formula}, but any $\spsp^{[\sqrt{n}]}$ circuit computing ${\cal Q}_n$ must have  top fan-in at least $2^{\Omega(\sqrt{n}\log N)}$. Moreover ${\cal Q}_n$ is computed by a polynomial size homogeneous $\spsp(r)$ formula for $r = \theta(\log n)$. More formally, we prove the following theorem.

 \begin{thm}[Depth reduction is tight for formulas]~\label{thm:dr}
There exists an explicit family of polynomials $\{{\cal Q}_n\}_n$ and an absolute constant $\epsilon > 0$ such that ${\cal Q}_n$ is of degree $\theta(n)$, in $N = \theta(n^2)$ variables, and computed by a $\poly(n)$ size homogeneous $\spsp(\log n)$ circuit (in particular a homogeneous arithmetic formula); and 
for every $t$ such that $\omega(\log n) \leq t \leq \epsilon n $, any $\spsp^{[t]}$ circuit computing ${\cal Q}_{n}$ must have top fan-in at least $\exp\left(\Omega(\frac{n}{t}\log N)\right)$.
\end{thm}
 
The above theorem follows by an interpolation argument applied to a {\it hierarchy} theorem for $\spsp^{[t]}$ circuits, which is the heart of our argument. The hierarchy theorem shows that by increasing the bound on the bottom fan-in of $\spsp^{[t]}$ circuits even slightly, we get a much richer class of arithmetic circuits. We believe this is an interesting result in its own right. 

\begin{thm}[Hierarchy theorem for $\spsp^{[t]}$ circuits]~\label{thm:hierarchy}
There exists an absolute constant $\epsilon > 0$ such that for every  $t$ with $\omega(\log n) \leq t \leq \epsilon n $, there exists an explicit family of polynomials $\{{\cal P}_{t,n}\}_n$ such that ${\cal P}_{t,n}$ is of degree $n$, has $N = n^2$ variables, and is computed by a $\poly(n)$ size homogeneous $\spsp^{[t]}$ circuit, and for every $t'$ s.t. $t' < t/20$, any homogeneous $\spsp^{[t']}$ circuit computing ${\cal P}_{t,n}$ must have top fan-in at least $\exp\left(\Omega(\frac{n}{t}\log N)\right)$.
\end{thm}

These results immediately imply that Koiran's and Tavenas' depth reduction~\cite{koiran, Tavenas13} is tight for formulas, for all but a small number of choices of the bottom fan-in. In particular, it is tight for the case where the bottom fan-in is bounded by $\sqrt{n}$. Interestingly enough, the polynomial size formulas computing ${\cal Q}_n$ are of {\dep}. In fact, they are a sum of $O(\log n)$ {\it regular} homogeneous formulas of {\dep}. 
 

A corollary of our results is that any conversion of a general (homogeneous) formula to a regular formula must incur a quasipolynomial blow up in size. It was shown in~\cite{KSS13} that any algebraic branching program can be converted to a regular formula with a quasipolynomial blow up in size. Since it is widely believed that formulas are much weaker than ABPs, it was conjectured that formulas, or homogeneous formulas might have a more efficient conversion (which would suffice in proving superpolynomial lower bounds for homogeneous formulas!). We show however that this is not true. 
Combining our results with the result of~\cite{KSS13}, we obtain the following (tight) lower bound for converting homogeneous formulas to regular formulas. 

\begin{thm}[Lower bounds for reduction to Regular Formulas]~\label{thm:regular}
There exists an explicit family of polynomials $\{{\cal Q}_n\}_n$ and an absolute constant $\epsilon > 0$ such that ${\cal Q}_n$ is of degree $\theta(n)$, in $N = \theta(n^2)$ variables, and computed by a $\poly(n)$ size homogeneous $\spsp(\log n)$ circuit (in particular a homogeneous arithmetic formula); and 
any regular formula computing ${\cal Q}_n$ must have size at least $N^{\Omega(\log n)}$. 
\end{thm}






\vspace{2mm}
\noindent 
{\bf Organization of the paper:} The rest of the paper is organized as follows. In Section~\ref{sec:prelims}, we introduce some preliminary notions about circuits and introduce notation which we will use in the rest of the paper. In Section~\ref{sec:lbproof} we prove our lower bound for homogeneous $\spsp(r)$ circuits when $r = o(\log n)$. In Section~\ref{sec:drtight}, we show that depth reduction is tight for homogeneous arithmetic formulas by showing it is tight for homogeneous $\spsp(\Omega(\log n))$ circuits. We conclude with some discussion and open problems in Section~\ref{sec:conclusion}.




\section{Preliminaries}~\label{sec:prelims}
\noindent
{\bf Arithmetic Circuits: }An arithmetic circuit over a field $\F$ and a set of variables $\overline{x} = \{x_1, x_2, \ldots, x_n\}$ is a directed acyclic graph such that every node in the graph is labelled by either a field element or a variable in $\overline{x}$ or one of the field operations $+, \times$. There could be one or more nodes with fan-out zero, called the output gates of the circuit. The nodes with fan-in zero indexed by the variables or field elements are called the leaf nodes. In  this paper, unless otherwise mentioned, we will assume that there is unique output gate. We will refer to the length of the longest path from an output node to a leaf node as the depth of the circuit. A circuit is said to be homogeneous if the polynomial computed at every node in the circuit is a homogeneous polynomial. A circuit is said to be a formula if the underlying undirected graph is  a tree. By a circuit of {\dep}, we will refer to a circuit of the form $\spsp$,  where the output gate is a $+$ gate and all the nodes at a distance $2$ from it are also labelled by $+$, and the remaining gates are labelled by $\times$. Observe that a {\it depth 4 circuit} can be converted into a {\it depth 4 formula} with only a polynomial blow up in size. We will therefore, use the term formula or circuit for a depth 4 circuit interchangeably in this paper. A homogeneous polynomial $P(\overline{x})$ of degree $d$ computed by a {\dep} circuit is of the form 

\begin{equation}\label{def:model}
P(\overline{x}) = \sum_{i=1}^{r}\prod_{j=1}^{d_i}{Q_{i,j}(\overline{x})}
\end{equation}

Based upon this definition, we will now define the specific restrictions of {\dep} circuits that we study in this paper. 

\vspace{2mm}
\noindent
{\bf Homogeneous $\spsp^{[a]}$ circuits and homogeneous $\sp^{[b]}\sp^{[a]}$ circuits:} The {\dep} $\spsp$ circuit in Equation~\ref{def:model}, is said to be a  $\sp^{[b]}\sp^{[a]}$ circuit, if 
each $Q_{i,j}(\overline{x})$ is a polynomial of degree at most $a$ and each $d_i$ is at most $b$. The {\dep} $\spsp$ circuit in Equation~\ref{def:model}, is said to be a  $\sp\sp^{[a]}$ circuit, if each $Q_{i,j}(\overline{x})$ is a polynomial of degree at most $a$. In this case we say that the {\it botton fan-in is bounded by $a$}. If the circuit is homogeneous, then we can assume without loss of generality that for each $i$, $\prod_{j=1}^{d_i}{Q_{i,j}(\overline{x})}$ is a polynomial of degree exactly $d$. 

Observe that, for each $i$, by grouping together and multiplying out some of the $Q_{i,j}$, we can transform a homogeneous $\spsp^{[a]}$ circuit into a homogeneous $\sp^{[b]}\sp^{[a]}$ circuit, where $b = O(\frac{d}{a})$. This operation of grouping together and multiplying would increase the size of the resulting circuit, but notice that it does not affect the top fan-in of the circuit. Thus lower bounds on the top fan-in for $\sp^{[O(b)]}\sp^{[a]}$ circuits impy the same lower bounds on the top fan-in of $\spsp^{[a]}$ circuits. 

\vspace{2mm}
\noindent
{\bf Homogeneous $\Sigma\Pi\Sigma\Pi(r)$ Circuits:} The {\dep} $\spsp$ circuit in Equation~\ref{def:model}, is said to be a $\Sigma\Pi\Sigma\Pi(r)$ circuit if the fan-in of the summation(top fan-in)  is bounded by $r$. Observe that there is no restriction on the bottom fan-in except that implied by the restriction of homogeneity. 

For each $i\in[r]$, the product $P_i = \prod_{j=1}^{d_i}{Q_{ij}}$ is said to be computed by the product gate $i$. Therefore, $P = \sum_{i=1}^{r}{P_i}$.
Here for every $i$ and $j$, $Q_{ij}$ is an $n$ variate homogeneous polynomial being computed by a $\Sigma\Pi$ circuit. The homogeneity restriction on $C$ implies 
that for every product gate $i$, 
\begin{equation}\label{eqn:homogdeg}
 \text{deg}(P) = d = \sum_{j=1}^{d_i}{\text{deg}(Q_{ij})}
\end{equation}
With every product gate $i \in [r]$, we can associate a multiset $(D_i,m_i)$, where 
\begin{equation}
D_i=\{\text{deg}(Q_{ij}): j\in[d_i]\} 
\end{equation}
and $m_i$ is a map from $D_i$ to $\N$, which assigns to every element $l$ in $D_i$, the number of $j\in[d_i]$ such that $Q_{ij}$ has degree equal to $l$. For a homogeneous depth $4$ circuit, computing a degree $d$ polynomial, Equation~\ref{eqn:homogdeg} can be rewritten as 
\begin{equation}\label{eqn:multiplicity}
 \text{deg}(P) = d = \sum_{l\in D_i}{l\times m_i(l)}
\end{equation}
 for each $i$ in $[r]$. $\spsp(r)$ circuits for which the multiset $(D_i, m_i)$ is the same for every product gate $i \in [r]$, are said to be $\spsp^*$ circuits. 


\vspace{2mm}
\noindent
{\bf Regular Formula:} The notion of regular formulas was introduced in~\cite{KSS13}, where superpolynomial lower bounds for this model were proved. 
\begin{define} A formula computing a degree $d$ polynomial in $n$ variables is said to be regular, if it satisfies the following conditions:
\begin{enumerate}
\item It has alternating layers of sum and product gates. 
\item All gates in a single layer have the same fan-in.
\item The formal degree of the formula is at most some constant multiple of  the degree of the polynomial being computed. 
\end{enumerate}

\end{define}

\vspace{2mm}
\noindent
{\bf Shifted Partial Derivatives: }The complexity measure used in showing lower bounds in this paper is the dimension of the shifted partial derivatives introduced in~\cite{Kayal12} and used in~\cite{FLMS13},~\cite{GKKS12} and~\cite{KSS13}.  For a field $\F$, an $n$ variate polynomial $P \in {{{\F}}}[\overline{x}]$ and a positive integer $k$, we denote by $\partial^{=k} P$, the set of all 
partial derivatives of order equal to $k$ of $P$. For a polynomial $P$ and a monomial $m$, we denote by ${\partial_m P}$ the partial derivative of $P$ with respect to $m$. Our proof uses the notion of shifted partial derivatives of a polynomial as defined below.  

\begin{define}[\cite{GKKS12}]\label{def:shiftedderivative}
For an $n$ variate polynomial $P \in {\field{F}}[\overline{x}]$ and integers $k, \ell \geq 0$, the space of $\ell$ shifted $k^{th}$ order 
partial derivatives of $P$ is defined as
\begin{align}
 \langle \partial^{=k} P\rangle_{\leq \ell} \stackrel{def}{=} \field{F}\mhyphen span\{\prod_{i\in [n]}{x_i}^{j_i}\cdot g  :  \sum_{i\in [n]}j_i \leq \ell, g \in \partial^{=k} P\}
\end{align}
\end{define}
\noindent

\vspace{2mm}
\noindent
{\bf Nisan-Wigderson Polynomials: }\footnote{ This definition is a slight variant of the definition in~\cite{KSS13}. We modify the definition to ensure homogeneity. It is not hard to see that all the bounds proved for the original polynomial in~\cite{KSS13} hold for this variant also.} We will now define the family of polynomials introduced in~\cite{KSS13}. These polynomials were used to prove improved lower bounds for homogeneous $\spsp^{[t]}$ circuits. (In an earlier paper by Gupta et al~\cite{GKKS12}, a slightly weaker lower bound was shown for the permanent.) For a prime power $n$, let $\F_n$ be a field of size $n$. For the set of $n^2$ variables $\{x_{i,j} : i, j \in [n]\} $ and $t \in [n]$, we define the degree $n$ homogeneous polynomial $NW_{t,n}$ as 

$$NW_{t,n} = \sum_{\substack{f(z) \in \F_n[z] \\
                        deg(f) < \lfloor \frac{n}{2t} \rfloor}} \prod_{i \in [n]} x_{i,f(i)}$$
Clearly, for every $n$ and $t$, $NW_{t,n}$ is in $\VNP$. The Nisan-Wigderson polynomial family $\{NW_{n}\}_n$ is a family of polynomials in $\VNP$ such that $NW_{n}$ is a polynomial of degree $n+1$ in $n^2+n$ variables $\{x_{i,j} : i, j \in [n] \} \cup \{y_i : i \in [n]\}$ defined as follows

$$NW_{n} = \sum_{i = 1}^n y_i\cdot NW_{i,n}$$

For more on arithmetic circuits, we refer the interested reader to the survey by Shpilka and Yehudayoff~\cite{SY10}.  

\section{Lower bounds for $\spsp(r)$ circuits, $r = o(\log n)$}\label{sec:lbproof}



In earlier works by Gupta et al~\cite{GKKS12} and Kayal et al~\cite{KSS13}, exponential lower bounds were shown for the class of homogeneous $\spsp$ circuits with bounded bottom fan-in.  Without the restriction on the bottom fan-in, basically no lower bounds for $\spsp$ circuits are known. In this section we prove the first super polynomial lower bounds for homogeneous$\spsp$ circuits with bounded top fan-in. 
The main technical core of our result is a {\it depth reduction} result, very similar in spirit to those by Koiran~\cite{koiran} and Tavenas~\cite{Tavenas13}. By exploiting the structure of these circuits, we show how   to get {\it improved depth reduction} for these circuits. The proof of our depth reduction is quite different from that of~\cite{koiran, Tavenas13}, and is somewhat subtle. We don't reduce to a $\spsp^{[t]}$ circuit, but a related and slightly more general class of circuits where instead of requiring an absolute bound on the bottom fan-in, we just require that in some sort of average sense, the bottom fan-in is small. In particular we reduce to a $\spsp$ circuit in which the sum of degrees of any $\epsilon n/t$ product gates at the bottom is at most $\epsilon n$. This is a more refined notion and a slightly more general class of circuits than $\spsp^{[t]}$ circuits. We observe that the shifted partial derivative technique does not distinguish between these two kinds of circuits, and thus we are still able to use a variant of Theorem~\ref{thm:gkks} to obtain our lower bounds. Thus in spirit we still get depth reduction. In fact everywhere in this paper we could replace $\spsp^{[t]}$ circuits with this slightly more general class of circuits, and none of the results would be affected.

 There seem to be two main obstacles in extending the lower bounds of~\cite{GKKS12, KSS13} to lower bounds for general  depth 4 homogeneous $\spsp$ circuits. The lower bounds in~\cite{GKKS12, KSS13} work only when the degrees of {\it all} polynomials feeding into the product gate at the second layer are small (in other words, the bottom fan in is small), say $\leq \sqrt n$.  It is easy to see that if the degrees of all polynomials feeding into the product gate at the second layer is large (i.e. the bottom fan-in of all the gates is large), say $\geq \sqrt n$, then for sparsity reasons and simple monomial counting, it is easy to obtain exponential lower bounds. The first obstacle is to handle the case when the degrees of some of the polynomials is small and for some of them it is large. For instance fix any arbitrary sequence $\mathcal D$ of degrees summing to $n$, and assume that the polynomials feeding into each product gate at the second from top layer have their degrees coming from this sequence. Is it still possible to obtain exponential lower bounds? The second obstacle to extending the results from~\cite{GKKS12} is to find a way to combine the lower bounds for all these various cases into a common lower bound for the case when the circuit is composed of product gates of different kinds. For instance we know lower bounds when all product gates at the second layer have small incoming degrees and when all product gates have large incoming degrees. However we do not know how to combine these lower bounds into a single lower bound when the circuit is the sum of two circuits, one of the low degree kind, and one of the high degree kind. In this paper we show how to resolve the first obstacle. Moreover when the top fan-in is $o(\log n)$, the second obstacle turns out to not be a problem either.

\paragraph{Proof Overview:} 
Most lower bounds for arithmetic circuits proceed by identifying some kind of ``progress measure", and show that for any given circuit in a circuit class, the measure is small if the size of the circuit is small, whereas for the polynomial one is trying to compute (for instance the permanent), the measure is large. In the results by Gupta et al.~\cite{GKKS12} and Kayal et al~\cite{KSS13}, the progress measure used is the dimension of the $\ell$ shifted $k^{th}$ order partial derivative $\dim(\langle \partial^{=k} P\rangle_{\leq \ell})$, for a suitable choice of $\ell$ and $k$. It is shown that every small depth $4$ circuit with bounded bottom fan-in has small $\dim(\langle \partial^{=k} P\rangle_{\leq \ell})$ compared to that of an explicit polynomial in $\VNP$, the $NW_n$ polynomial.  Thus if a depth $4$ circuit with bounded bottom fan-in must compute $NW_n$, then it must be large.  More precisely it is shown that every product gate $Q_i = \prod_{j = 1}^d Q_{ij}$ has $\dim(\langle \partial^{=k} P\rangle_{\leq \ell})$ much smaller than that of the permanent, provided the degrees of the $Q_{ij}$ are small. This is the core of the argument. Combined with the sub-additivity of $\dim(\langle \partial^{=k} P\rangle_{\leq \ell})$, the result easily follows. 

Our proof builds upon the results of~\cite{GKKS12} and~\cite{KSS13}, and combines the use of the progress measure $\dim(\langle \partial^{=k} P\rangle_{\leq \ell})$ with the notion of ``sparsity" to prove our improved depth reduction and the lower bounds for the polynomial family $\{NW_n\}_n$. 
Suppose $C = \sum_{i=1}^{r}\prod_{j=1}^{d_i}{Q_{ij}}$ is a homogeneous $\spsp$ circuit computing $NW_n$. If all the $Q_{ij}$ had low degree, then the results of~\cite{GKKS12} and~\cite{KSS13} give exponential lower bounds for the top fan-in of $C$. Also in the extreme case where all the $Q_{ij}$ have high degree, then since $C$ is homogeneous, the number of $Q_{ij}$ per product gate $Q_i = \prod_{j = 1}^d Q_{ij}$ must be small, and hence their product cannot have too many monomials\footnote{The number of monomials in each $Q_{ij}$ is a most the size of the circuit.}. If the number of monomials is too few, we would not even be able to get all the monomials in $NW_n$. In general, of course there might be some high degree and some low degree polynomials, and we attempt to interpolate between the two settings to obtain our results.   

For each product gate $Q_i = \prod_{j = 1}^{d_i} Q_{ij}$, recall that each $Q_{ij}$ is a homogeneous polynomial of degree $d_{ij}$ (say), and $\sum_{j = 1}^{d} d_{ij} = n$. If the size of the circuit is at most $s$, then each $Q_{ij}$ has at most $s$ monomials. We decompose each product gate into its inputs $Q_{ij}$ of {\it high degree}  (those of degree $\geq t$) and its inputs $Q_{ij}$ of {\it low degree} (those of degree $< $t). Observe that there cannot be too many (greater than $n/t$) high degree polynomials $Q_{ij}$ as otherwise their product would have degree
exceeding $n$. Thus the product of all the high degree $Q_{ij}$ cannot have more than $s^{n/t} $ monomials. Let $H$ be the product of the the high degree $Q_{ij}$, and $L$ be the product of the low degree $Q_{ij}$. Then, by writing out $H$ as a sum of monomials ($H = \sum_k h_k$) and multiplying each monomial $h_k$ with $L$, we can expand out $Q$ as $\sum_k h_k\cdot L$. Note that $L$ is a product of low degree polynomials. Also, each $h_k$ is a monomial and hence a product of degree $1$ polynomials. Thus we have expressed $Q$ as a $\spsp^{[t]}$ circuit, where now all the product gates multiply polynomials of degree at most $t$. 

The hope at this point would be to apply this transformation to all the product gates and then possibly apply the result in~\cite{KSS13}  to obtain a lower bound. The trouble with this argument is that under the transformation described, the top fan-in of the original circuit might blow up by a factor equaling the number of monomials in $H$, which could be nearly as large as $s^{n/t}$. With this loss in parameters, the bound given by the \cite{KSS13} result gives nothing nontrivial. Thus in general one cannot choose an absolute threshold $t$ and for all product gates choose degrees greater than $t$ to be the high degree polynomials and the ones below $t$ to the the low degree polynomials. 

What we show is that by examining the degrees of the polynomials feeding into the product gates, one can carefully choose a threshold $t$ that works for each product gate individually, though it might not be the same threshold for all gates. It turns out that this threshold that we find is purely a function of the degree sequence $\mathcal D$ of the product gate. Thus if all product gates have the {\it same} degree sequence, i.e. we have a $\Sigma\Pi^{\mathcal D}\Sigma\Pi$ circuit, then we obtain exponential lower bounds. However, for general $\spsp$ circuits it can be a problem, since if the threshold is different for different gates, we do not have any one single progress measure that works for all gates and thus for the entire circuit. However we are still able to show that for each gate, only very few thresholds are ``bad", and when the top fan-in is $o(\log n)$, then we show there is a single threshold that will work for all gates to give superpolynomial lower bounds. 

\subsection{Proof of Theorem~\ref{thm:lb1}}
In this subsection, we will present the proof of Theorem~\ref{thm:lb1}. Let us consider a homogeneous $\spsp(r)$ circuit $C$ of size $s$ computing ${NW}_n$. From Equation~\ref{def:model}, this implies that 
\begin{equation}
 NW_n = \sum_{i=1}^r\prod_{j=1}^{d_i}Q_{ij}
\end{equation}
 where for every value of $i$ and $j$, $Q_{ij}$ is a homogeneous polynomial being computed by a subcircuit of depth $2$ of $C$. Observe that
$Q_{ij}$ is being computed by a $\sp$ circuit and hence, the number of monomials with nonzero coefficients in a sum of products expansion of $Q_{ij}$ will 
be at most the size of $C$. In other words, $Q_{ij}$ is $s$ sparse for each $i\in[r]$ and $j\in[d_i]$. Without loss of generality, we will assume that 
for every $i\in[r]$, $d_i = n$, since if $d_i<n$ for any $i$, we can always make it equal to $n$ adding dummy polynomials that are the constant $1$. 

Let us now consider the polynomial computed at a product gate near the top of $C$. It is of the form $Q= \prod_{i\in[n]} Q_i$. 
Let us also assume without loss of generality that the $Q_i$ are arranged in non-increasing order of their degrees. 
The idea of the proof, as described in the overview,  would be to decompose the $Q_i$ into {\it high degree} and {\it low degree} parts 
and then multiply out all the {\it high degree} parts and  count on their
sparsity to show that the product does not blow up the dimension of the space of shifted partial derivatives by too much. We will then use the following lemma implicit in the work of~\cite{GKKS12}, to obtain our bounds.

\begin{lem}[Implicit in~\cite{GKKS12}]~\label{lem:lowdeg}
Let $P = \prod_{i = 1}^d \tilde{P_i}$ be a polynomial in $N$ variables such that the sum of the degrees of any $k$ of these $d$ polynomials $\tilde{P_1}, \tilde{P_2}, \ldots, \tilde{P_d}$ is at most $D$. Then, for every integer $\ell \geq 0$,
$$\dim(\langle \partial^{=k} f\rangle_{\leq \ell}) \leq {d+ k-1 \choose k}{N + D-k + \ell \choose N}.$$
\end{lem}

\begin{proof}
The proof of the lemma is exactly the same calculation as in~\cite{GKKS12}. We replace their bound of $tk$ (which for them was the sum of degrees of $k$ polynomials of degree at most $t$), by our bound of  $D$. 
\end{proof}

The following lemma is the core of our argument. 

\begin{lem}\label{lem:sps}
Let $Q = \prod_{j\in [n]} Q_i$ be a depth 3 $\Pi\Sigma\Pi$ homogeneous circuit of degree $n$ in $N$ variables, where each $Q_i$ has at most $s$ monomials. Let $0< \epsilon < 1$ be any small constant. 
Consider $k = n^{i/m}$, for $1 \leq i \leq m$ and any integer $\ell \geq 0$. Then for all but $1/\epsilon$ choices of $i$ , 
$$\dim(\langle \partial^{=k} Q\rangle_{\leq \ell}) \leq s^{k\cdot n^{-1/m}}\cdot {k/\epsilon+ k-1 \choose k}{N + 4\epsilon n-k + \ell \choose N}.$$
\end{lem}

\begin{proof}
Since the $Q_i$'s are arranged in order of decreasing degree, $Q_1$ has highest degree and ${\cal Q}_n$ has the smallest degree. 

For $1 \leq i \leq m$, let $S_i = \{Q_j | j \leq n^{i/m}\} $ be the set of the first $n^{i/m}$ of the  $Q_j$'s. For each $i$, we will sum the degrees of the $Q_j$'s in $S_i \setminus S_{i-1}$. Let $$D_i = \sum_{j \mbox{ s.t.} Q_j \in S_i \setminus S_{i-1} } \deg(Q_j).$$

Then $\sum_{i=1}^{m} D_i = n$. Thus there are at most $1/\epsilon$ choices of $i$ for which $D_i \geq \epsilon n$. We will show that for all other choices of $i$, for $k = n^{i/m}$ and any integer $\ell \geq 0$, $\dim(\langle \partial^{=k} Q\rangle_{\leq \ell}) \leq s^{k\cdot n^{-1/m}}\cdot {k/\epsilon+ k-1 \choose k}{N + 4\epsilon n-k + \ell \choose N}.$

Let us fix $i$ such that $D_i \leq \epsilon n$. We will split up the various $Q_j$'s into those that are in $S_{i-1}$ and those that are not. For those $Q_j$ in $S_{i-1}$, we will exploit the fact that there aren't too many of them and they each have at most $s$ monomials, to show that they do not affect the dimension of shifted partial derivatives by too much. For the rest of the $Q_j$ we will take advantage of the fact that their degrees are not too large, and hence the sum of degrees of any $k$ of them is small, and thus we will be able to bound the span of shifted partial derivatives of their product using the argument presented in~\cite{GKKS12}. 

Let $H= \prod_{Q_j \in S_{i-1}} Q_j$, and let $Q_{\bar H} = Q/H$.  Since each $Q_i$ has at most $s$ monomials, thus $H$ has at most $s^{n^{(i-1)/m}}$ monomials. Hence we can express the polynomial $Q$ as the sum of at most $s^{n^{(i-1)/m}}$ polynomials $P_1, P_2, \ldots P_u$, where each of the polynomials is the product of some monomial (from $H$), and the product of all the $Q_j$ that are not in  $S_{i-1}$ (i.e. those in $Q_{\bar H}$) .

We will show that for each $P_j$, $1\leq j \leq u$, for $k = n^{i/m}$, $$\dim(\langle \partial^{=k} P_j\rangle_{\leq \ell}) \leq {k/\epsilon + k-1 \choose k}{N + 4\epsilon n-k + \ell \choose N}.$$ Since $u$ is at most the number of monomials in $H$, thus $u \leq s^{n^{(i-1)/m}} = s^{k\cdot n^{-1/m}}$. 
Since $Q = \sum_{j \in [u]} P_j$, the sub-additivity of $\dim(\langle \partial^{=k} \rangle_{\leq \ell})$ will imply that 
$$\dim(\langle \partial^{=k} Q\rangle_{\leq \ell}) \leq s^{k\cdot n^{-1/m}}\cdot {k/\epsilon+ k-1 \choose k}{N + 4\epsilon n-k + \ell \choose N}.$$

Let us focus our attention on any one of these polynomials $P_j$, and call it $P$.

Then $P = h\cdot Q_{\bar H} = h\cdot\prod_{j \geq n^{(i-1)/m} + 1 } Q_j$, where $h$ is a monomial of $H$ and can be thus written as a product of degree one homogeneous polynomials. Let us rename the degree $1$ polynomials in $h$ and the different $Q_j$ dividing $Q_{\bar H}$, so that $P = \hat P_{1}\hat P_{2}\cdots \hat P_{\ell}$. 

Consider all the polynomials $\hat P_i$ dividing $P$ which have degree at most $\epsilon n/k$, and group them together and multiply them so that each of the grouped polynomials now has degree at least $\epsilon n/k$ and at most $2\epsilon n/k$. Clearly this can be done. Call the new set of polynomials (the grouped ones and the ones that had degree at least $\epsilon n/k$ to start out with) $\tilde{P_1}, \tilde{P_2}, \ldots, \tilde{P_d}$. Since the sum of their degrees is at most $n$, thus the total number $d$ of these polynomials is at most $ k/\epsilon$.

\begin{prop}
The sum of the degrees of any $k$ of these $d$ polynomials $\tilde{P_1}, \tilde{P_2}, \ldots, \tilde{P_d}$ is at most $4\epsilon n$. 
\end{prop}
\begin{proof}
Out of the $k$ polynomials, we see what fraction lie among the ``grouped" polynomials, and what lie among the original ungrouped polynomials. 
Recall that by the choice of $i$, and setting $k = n^{\frac{i}{m}}$, the sum of degrees of any $k-kn^{\frac{-1}{m}}$ of the $\hat P_{i}$ dividing $P$ was 
at most $\epsilon n$. Since $m$ is $o(\log n)$, the sum of the degrees of any $k$ of them will be at most $2\epsilon n$. Thus, the contribution from the original ungrouped polynomials is at most $2\epsilon n$. 
Also, the contribution from the grouped polynomials can be at most $2\epsilon n$ since there are at most $k$ of them, and each has degree at 
most  $2\epsilon n/k$.  Thus the total sum of degrees is at most $4\epsilon n$.
\end{proof}

Thus, $P= \prod_{i = 1}^d\tilde{P_i}$ is a polynomial in $N$ variables such that the sum of the degrees of any $k$ of the $d$ polynomials $\tilde{P_1}, \tilde{P_2}, \ldots, \tilde{P_d}$ is at most 
$D = 4\epsilon n$. Recall also that $d \leq k/\epsilon$. Hence, by Lemma~\ref{lem:lowdeg}, for any integer $\ell \geq 0$,
$$\dim(\langle \partial^{=k} P\rangle_{\leq \ell}) \leq {k/\epsilon + k-1 \choose k}{N + 4 \epsilon n-k + \ell \choose N}.$$


 



\end{proof}

\begin{thm}\label{thm:choosingk}
Let $C$ be a homogeneous $\Sigma\Pi\Sigma\Pi(r)$ circuit in $N$ variables, of size $s$ and of degree at most $n$. Then for all constants $\epsilon$, with $0< \epsilon < 1$, there exists a choice of $i$, with  $1 \leq i \leq 2r/\epsilon$, such that for $k = n^{\epsilon i/2r} $, 
and for all integers $\ell \geq 0$,  $$\dim(\langle \partial^{=k} C \rangle_{\leq \ell}) \leq r\cdot s^{k\cdot n^{-\epsilon/2r}}\cdot {k/\epsilon+ k-1 \choose k}{N + 4\epsilon n-k + \ell \choose N}.$$
\end{thm}

\begin{proof}
Let $m = 2r/\epsilon$. Let $C = \sum_{j = 1}^r Q_j$. Let $i \in [m]$.

Then for each $Q_j$, by Lemma~\ref{lem:sps}, for all but $1/\epsilon$ choices of $i$, for $k = n^{i/m}$, 
$$\dim(\langle \partial^{=k} Q_j\rangle_{\leq \ell}) \leq s^{k\cdot n^{-1/m}}\cdot {k/\epsilon+ k-1 \choose k}{N + 4\epsilon n-k + \ell \choose N}.$$

Hence for each $Q_j$ we get at most $1/\epsilon$ choices of $i$ that may not work to get the bound above, and we call those choices ``bad" for $Q_j$.  We call the rest of the choices ``good" for $Q_j$.  Thus by the union bound there are at most $r/\epsilon$ choices of $i$ that are bad for some $Q_j$. Since $m > r/\epsilon$, thus there is a choice of $i \in [m]$ that is good for every $Q_j$.

Thus for any integer $\ell \geq 0$ and $k = n^{i/m}$, for all $j \in [r]$, $$\dim(\langle \partial^{=k} Q_j \rangle_{\leq \ell}) \leq s^{k\cdot n^{-1/m}}\cdot {k/\epsilon+ k-1 \choose k}{N + 4\epsilon n-k + \ell \choose N}.$$
 Hence  $$\dim(\langle \partial^{=k} C \rangle_{\leq \ell}) \leq r\cdot s^{k\cdot n^{-1/m}}\cdot {k/\epsilon+ k-1 \choose k}{N + 4\epsilon n-k + \ell \choose N}.$$
\end{proof}

We can observe that the choice of the threshold and $k$ for every product gate just depends upon the multiset of the degrees associated with the input feeding into it. In particular, if we start with a $\spsp^*$ circuit, then the value of the threshold and $k$ that works for one product gate also works for the circuit in general. Hence, we have the following theorem which gives us an upper bound on the dimension of the shifted partial derivative space of a $\spsp^*$ circuit.

\begin{thm}\label{thm:ubspspstar}
Let $C$ be a homogeneous $\Sigma\Pi\Sigma\Pi^*$ circuit in $N$ variables, of size $s$, top fan-in $r$ and of degree at most $n$. Then for all constants $\epsilon$, with $0< \epsilon < 1$, there exists a choice of $i$, with  $1 \leq i \leq m$, where $m = 1/\epsilon + 1$ such that for $k = n^{i/m} $, 
and for all integers $\ell \geq 0$,  $$\dim(\langle \partial^{=k} C \rangle_{\leq \ell}) \leq r\cdot s^{k\cdot n^{-1/m}}\cdot {k/\epsilon+ k-1 \choose k}{N + 4\epsilon n-k + \ell \choose N}.$$
\end{thm}
It is important to note the difference between the bounds in Theorem~\ref{thm:choosingk} and Theorem~\ref{thm:ubspspstar}. In Theorem~\ref{thm:ubspspstar}, the exponent of $s$ is independent of the top fan-in $r$ as $m$ is a constant.

In order to complete the proof now, we will look at the shifted partial derivative complexity of the circuit as well as of the polynomial $NW_n$ under restrictions where all the variables $\{y_1, y_2, \ldots, y_n\}$ are set to constants. The partial derivatives as well the final shifts are just taken with respect to monomials in the $n^2$ variables $\{x_{1,1}, x_{1,2}, \ldots, x_{n,n}\}$. The following theorem tells us that under some restrictions of this type, $NW_n$ has large complexity. This happens because under the restriction where  $y_t = 1$ and $y_j = 0$ for every $j\neq t$, we obtain $NW_{t,n}$ from $NW_n$.  
\begin{thm}[\cite{KSS13}]\label{thm:gkss-perm}For any integers $t$, $k$, $\ell$ such that $\log^2 n \leq t \leq \frac{n}{100}$, $k = \lfloor \frac{n}{2t} \rfloor$, and $\ell = \lceil \frac{5n^2t}{\log n}\rceil$,
$$\dim(\langle \partial^{=k} NW_{t,n} \rangle_{\leq \ell}) \geq  \frac{1}{n^3}{{n^2 + \ell + n - k} \choose n^2}$$
\end{thm}

We will also use the following result from~\cite{KSS13} in our calculations.
\begin{thm}[\cite{KSS13}]\label{thm: gkss-calc}
For any fixed constant $\alpha$ and $t$, $k$, $\ell$ such that $\log^2 n \leq t \leq \frac{n}{100}$, $k = \lfloor \frac{n}{2t} \rfloor$ and $\ell = \lceil \frac{5n^2t}{\log n}\rceil$, if 
$$E = \frac{\frac{1}{n^3}{{n^2 + \ell + n - k} \choose n^2}}{{\frac{\alpha n}{t} \choose k}{{n^2 + \ell + k(t-1)}\choose{n^2} }}$$
Then, $E \geq \exp(\Omega(\frac{n}{t}\log n))$.
\end{thm}
For the range of values of $t$ stated above, the value of $k$ lies in the range $200 \leq k \leq \frac{n}{2\log^2n}$. To complete the proof, we will argue that after setting the $y$ variables to a constant, there is a value of $k$ in this range and an $\ell$ such that the dimension of the shifted partial derivative span of the circuit is small . Based on this value of $k$, we will then invoke a particular projection $NW_{t,n}$ of $NW_n$ and then use the bound from Theorem~\ref{thm:gkss-perm}.  

\begin{proof}[Proof of Theorem~\ref{thm:lb1}]
Let us consider a $\spsp(r)$ circuit of size $s$ which computes the polynomial $NW_n$. As discussed, we will analyze the shifted partial derivative complexity of the circuit and the polynomial under the restriction that the $\{y_1, y_2, \ldots, y_n\}$ variables are set to constants. So, the degree of the polynomial computed is $n$ and the number of alive variables is $N = n^2$.
Let $0<\epsilon<1$ be a constant. We will now show that we can choose a value of $k$ such that the conditions in Theorem~\ref{thm:choosingk} and Theorem~\ref{thm:gkss-perm} hold. From the proof of Theorem~\ref{thm:choosingk}, we know that there are at most $\frac{r}{\epsilon}$ many choices of integer
$0<i<\frac{2r}{\epsilon}$ that are bad i.e that $k = n^{\frac{\epsilon\cdot i}{2r}}$ does not give us the upper bound on the complexity of the shifted partial derivatives as stated in Theorem~\ref{thm:choosingk}. Now, all we need to show is that there is one such ``good'' $i$ such that $200 \leq k = n^{\frac{\epsilon\cdot i}{2r}} \leq \frac{n}{2\log^2 n}$. For this to hold, we need to show a ``good'' $i$ in the range $\frac{2r}{\epsilon\log n}\log 200 < i < \frac{2r}{\epsilon}(1-\frac{1+2\log\log n}{\log n})$. The number of integers in this range is at least $\frac{2r}{\epsilon}(1-3\frac{\log\log n}{\log n})$, while the number of bad $i$ is at most $\frac{r}{\epsilon}$. Hence, for $n$ large enough, there is an $i$ such that for the resulting $k$, the bound in Theorem~\ref{thm:choosingk} holds and $t = \frac{n}{2k}$ satisfies $\log^2 n \leq t \leq \frac{n}{100}$. 
Let us fix such a good $k$. Let us now fix $t = \frac{n}{2k}$, $\ell = \frac{5n^2t}{\log n}$ and $\epsilon = \frac{1}{8}$. Now, let us consider the restriction of $C$ when just $y_t$ is set to $1$ and $y_j$ is set to zero for every $j\neq t$. In this case, the circuit just computes $NW_{t,n}$. From Theorem~\ref{thm:choosingk}, we get 
$$\dim(\langle \partial^{=k} C \rangle_{\leq \ell}) \leq r\cdot s^{k\cdot n^{-\epsilon/2r}}\cdot {k/\epsilon+ k-1 \choose k}{N + 4\epsilon n-k + \ell \choose N} $$ and from Theorem~\ref{thm:gkss-perm}, we get 
 $$\dim(\langle \partial^{=k} NW_{t,n} \rangle_{\leq \ell}) \geq  \frac{1}{n^3}{{n^2 +\ell+ n - k} \choose n^2}$$ 
So, if  $C$ computes $NW_n$, then $\dim(\langle \partial^{=k} C \rangle_{\leq \ell}) \geq \dim(\langle \partial^{=k} NW_{t,n} \rangle_{\leq \ell})  $. 
Thus $$r\cdot s^{k\cdot n^{-\epsilon/2r}}\cdot {k/\epsilon+ k-1 \choose k}{N + 4\epsilon n-k + \ell \choose N} \geq \frac{1}{n^3}{{n^2 +\ell+ n - k} \choose n^2}.$$
Substituting the parameters, we get 
$$ r\cdot s^{k\cdot n^{-1/16r}} \geq \frac{\frac{1}{n^3}{{n^2 +\ell+ n - k} \choose n^2}}{{\frac{\alpha n}{t} \choose k}{{n^2 +\ell+ k(t-1)}\choose{n^2} }}$$ for some appropriate constant $\alpha$ dependent on $\epsilon$. From theorem~\ref{thm: gkss-calc}, this implies that 
$$r\cdot s^{k\cdot n^{-1/16r}} \geq \exp(\Omega(\frac{n}{t}\log n)) = \exp(\Omega(k\log n)) $$
Using the fact that $r$ is at most $s$ (in fact it is much much smaller), we conclude that 
$$k\cdot n^{-1/16r}\cdot \log s \geq \Omega(k\log n).$$
Thus $$\log s \geq \Omega(n^{1/16r}\log n)$$ and hence $$s \geq \exp\left(n^{\Omega(1/r)}\log n\right).$$

\end{proof}

A very similar calculation lets us prove Theorem~\ref{thm:lb2}.

\begin{proof}[Proof of Theorem~\ref{thm:lb2}]

For a $\spsp^*$ circuit, the calculation will proceed exactly the same as above, and in the end, we will get 
$$s \geq \exp\left(n^{\Omega(1/m)}\right),$$ which on substituting $m = 1/\epsilon + 1$, completes the proof. 
Thus, we obtain exponential lower bounds for $\spsp^*$ circuits computing the polynomial $NW_n$ regardless of their top fan-in.
\end{proof}




 \section{Depth reduction is tight for $\spsp(\Omega (\log n))$ circuits}\label{sec:drtight}
In this section, we will show that the depth reduction procedure of Koiran and Tavenas~\cite{koiran, Tavenas13} as given in Theorem~\ref{thm:tavenas} is tight. On the way to this result, we will prove a {\it Hierarchy} theorem(Theorem~\ref{thm:hierarchy}) for formulas of depth 4 with bounded bottom fan-in. We will then build up on this proof, and prove Theorem~\ref{thm:dr} and Theorem~\ref{thm:regular}. We will first provide an overview of the proof. 

\paragraph{Proof Overview:}We will construct an infinite family of polynomials $\{{\cal Q}_n\}_n$  (here $n$ is a prime power), such that ${\cal Q}_n$ is a homogeneous polynomial in $N = \theta(n^2)$ variables of degree $n+1$ which can be computed by a polynomial sized homogeneous $\spsp(O(\log n))$ circuit. We will show that for each $\omega(\log n) \leq a \leq \frac{n}{800}$, ${\cal Q}_n$ requires homogeneous $\hsp$ circuits of top fan-in $2^{\Omega(\frac{n}{a}\log n)}$. In order to construct this polynomial family, we will construct for each $\omega(\log n)\leq t\leq \frac{n}{40}$, a family of polynomials $\{{\cal P}_{t,n}(\overline{x})\}$, such that each ${\cal P}_{t,n}(\overline{x})$ is a homogeneous polynomial in $n^2$ variables and of degree $n$, and can be computed by a polynomial sized homogeneous $\spsp^{[t]}$ circuit. Moreover, we will show that any homogeneous $\spsp^{[t/20]}$ circuit computing it must have  top fan-in at least $2^{\Omega(\frac{n}{t}\log n)}$. We will then apply the interpolation trick of~\cite{KSS13} to ${\cal P}_{t,n}$ for various $t$ to obtain the ${\cal Q}_n$.  The construction is heavily inspired by the idea of constructing hard polynomials using Nisan-Wigderson designs used in~\cite{KSS13}. To show the lower bound for each $t$, we will use ideas from~\cite{CM13} and~\cite{FLMS13}, and show that for suitable $k$, $\partial^{=k}(P(\overline{x}))$ has a large number of elements whose leading monomials are at a ``large distance'' from each other.

\subsection{Proof of Theorem~\ref{thm:hierarchy}}
For the rest of this section, we will assume that $n$ is a prime power. For each such $n$, we will identify the elements of the field ${\F}_n$ with the elements of the set  $[n] = \{1,2,\ldots, n\}$. For a parameter $t$ which is a positive integer less than $n$, let us now partition the set $[n]$ into $\lceil \frac{n}{t} \rceil$ parts which are roughly equal and each is of size about $t$. For brevity, we will indicate $\frac{n}{t}$ by $\tilde{t}$. We will let $C_i = \{t(i-1)+1, t(i-1) + 2, \ldots, ti\}$ denote the $i^{th}$ such partition.  Also, for every $j \in [\tilde{t}]$ and $i \leq t$, let $C_j^{i}$ be the set of the $i$ smallest elements in $C_j$. Let us also consider a parameter $p$ which we will later set to an appropriately chosen constant. Let ${{\cal S}_p}$ be the set of all univariate polynomials of degree $p$ over $\F$ and let ${{\cal S}_p}^{\tilde{t}}$ be the set of ordered $\tilde{t}$ tuples over ${{\cal S}_p}$. Clearly, $|{{\cal S}_p}|$ is $\theta(n^{p+1})$, when $p$ is a constant. In the rest of the paper, we will use $\overline{x}$ to denote the set of $n^2$ variables $\{x_{i,j} : i, j \in [n]\}$ and $\overline{y}$ to denote the set of variables $\{y_1, y_2, \ldots, y_n\}$. We will use the following notion of distance between two monomials as defined in~\cite{CM13}. 
\begin{define}[\cite{CM13}]
Let $m_1$ and $m_2$ be two monomials over a set of variables. Let $S_1$ and $S_2$ be the multiset of variables in $m_1$ and $m_2$ respectively, then the distance $\Delta(m_1, m_2)$ between $m_1$ and $m_2$ is the min$\{|S_1| - |S_1\cap S_2|, |S_2| - |S_1\cap S_2|\}$ where the cardinalities are the order of the multisets.   
\end{define} 

We will also use the following notion of distance between strings or ordered tuples. For any two strings $s_1, s_2$ of the same length, the distance between them $\Delta(s_1, s_2)$ is the number of coordinates at which $s_1, s_2$ disagree with each other. For brevity, we will use $\alpha m$ to refer to $\lfloor \alpha m \rfloor$ for any positive integer $m$ and any real number $\alpha$.  

Based on the notations defined, we define the class of polynomials ${\cal P}_{p,t,n}$:

$${\cal P}_{p,t,n}(\overline{x}) = \prod_{j\in [\tilde{t}]}\sum_{f \in {{\cal S}_p}}\prod_{i \in C_j}x_{i, f(i)}$$

From the expression above, it follows that for every $n, t$ and a constant $p$, ${\cal P}_{p,t,n}$ can be computed by a polynomial sized $\Pi\Sigma\Pi$ formula. Observe that in fact it can be computed by a {\it regular} formula\footnote{When $t$ divides $n$ the formula will be exactly regular, and if not a simple modifcation could make it regular, but the details are simple and irrelevant.}. We summarize this observation below.

\begin{obs}~\label{obs:upperbound}
For every $n$ and a constant $p$, ${\cal P}_{p,t,n}$ can be computed by a $\Pi\Sigma\Pi$ regular formula of size polynomial in $n$.
\end{obs} 

We now intend to use the setup introduced in~\cite{CM13} by Chillara and Mukhopadhyay to show that this polynomial requires  homogeneous $\hspn$ circuits of top fan-in at least $2^{\Omega(\frac{n}{t}\log n)}$. This forms the basis for our hierarchy theorem. In~\cite{CM13} the following theorem is proved, which gives a sufficient condition to show hardness for homogeneous $\spsp^{[\sqrt{n}]}$ circuits. 

\begin{thm}[Theorem 3 in~\cite{CM13}]
Let $f(X)$ be a polynomial of degree $n$ in $n^{O(1)}$ variables such that for some constant $\delta$ there are $n^{\delta k}$ different polynomials in $\partial^{=k} f$ for $k = \gamma\sqrt{n}$(where $0<\gamma<1$ is a constant) such that any two of their leading monomials have distance at least $d = \frac{n}{c}$ for a constant $c>1$. Then, any homogeneous $\spsp^{[\sqrt{n}]}$ circuit that computes $f(X)$ must have top fan-in at least $2^{\Omega(\sqrt{n}\log n)}$. 
\end{thm}

Although the result is stated for $k$ being around $\sqrt{n}$, the theorem is also true for $k$ in a larger range of values of $k$. To show our lower bounds, we will argue that there are a ``large'' number of $k^{th}$ order partial derivatives of ${\cal P}_{p,t,n}$ for some appropriate $k$, whose leading monomials have the distance property stated above. 

To this end, we will now try and understand the monomial structure of the partial derivatives of an appropriately chosen order of ${\cal P}_{p,t,n}$. Now, from the definition of ${\cal P}_{p,t,n}$, every monomial in it can be identified by an ordered tuple of length $\tilde{t}$ over the set of polynomials in ${{\cal S}_p}$ and vice versa. So, for any $\overline{f} = (f_1, f_2, \ldots, f_{\tilde{t}}) \in {{\cal S}_p}^{\tilde{t}}$, let 
$$m_{\overline{f}} = \prod_{j\in [\tilde{t}]}\prod_{i \in C_j}x_{i, f_j(i)}$$ 

From the definitions above and that of ${\cal P}_{p,t,n}(\overline{x})$, it follows that 
$${\cal P}_{p,t,n}(\overline{x}) = \sum_{{\overline{f} \in {{\cal S}_p}^{\tilde{t}}}}m_{\overline{f}}$$

Let 
$$m_{\overline{f}}^{'} = \prod_{j\in [\tilde{t}]}\prod_{i \in C_j^{2p}}x_{i, f_j(i)}$$

When we finally set parameters, we will always have $p$ is a constant while $t$ increases with $n$. So, for $n$ large enough, $m_{\overline{f}}^{'}$ divides $m_{\overline{f}}$ and 

$$\frac{m_{\overline{f}}}{m_{\overline{f}}^{'}} = \prod_{j\in [\tilde{t}]}\prod_{i \in C_j\setminus C_j^{2p}}x_{i, f_j(i)}$$

Now, we set $k = 2p\tilde{t}$ and look at the partial derivatives of  ${\cal P}_{p,t,n}$ of order $k$.  For each $\overline{f} \in {{\cal S}_p}^{\tilde{t}}$, the degree of $m_{\overline{f}}^{'}$ equals $k$. Hence, $\partial^{=k} {\cal P}_{p,t,n}$ includes the set of partial derivatives of ${\cal P}_{p,t,n}$ with respect to $m_{\overline{f}}^{'}$ for each $\overline{f} \in {{\cal S}_p}^{\tilde{t}}$. From the definition of $m_{\overline{f}}$ and $m_{\overline{f}}^{'}$, and the fact that each polynomial in ${\cal S}_p$ has degree equal to $p$ and two distinct polynomials in ${\cal S}_p$ cannot agree on more than $p$ points, for any $\overline{f} \in {{\cal S}_p}^{\tilde{t}}$ and $\overline{g} \in {{\cal S}_p}^{\tilde{t}}$,

$$
{\partial_{m_{\overline{f}}^{'}} m_{\overline{g}}} = \left\{
        \begin{array}{ll}
            0 & \overline{f} \neq \overline{g} \\
           \frac{m_{\overline{f}}}{m_{\overline{f}}^{'}} & \overline{f} = \overline{g}
        \end{array}
    \right.
$$

From this discussion, the following lemma follows. 

\begin{lem}
For every $\overline{f} \in {{\cal S}_p}^{\tilde{t}}$, ${\partial_{m_{\overline{f}}^{'}} {\cal P}_{p,t,n}}$ is a monomial and equals $ \frac{m_{\overline{f}}}{m_{\overline{f}}^{'}}$. 
\end{lem}

At this point, we might hope to argue that for each $\overline{f} \in {{\cal S}_p}^{\tilde{t}}$ and $\overline{g} \in {{\cal S}_p}^{\tilde{t}}$ such that $\overline{f} \neq \overline{g}$, 
the distance between the monomials $\frac{m_{\overline{f}}}{m_{\overline{f}}^{'}}$ and $\frac{m_{\overline{g}}}{m_{\overline{g}}^{'}}$ is large. This statement in itself is not true, for if $\overline{f}$ and $\overline{g}$ differ in just one coordinate, then the distance between $\frac{m_{\overline{f}}}{m_{\overline{f}}^{'}}$ and $\frac{m_{\overline{g}}}{m_{\overline{g}}^{'}}$ could be as small as $t-3p$, which as it turns out is insufficient to achieve the desired bounds. Observe that if $\overline{f}$ and $\overline{g}$ differ in $i$ coordinates, then the distance between $\frac{m_{\overline{f}}}{m_{\overline{f}}^{'}}$ and $\frac{m_{\overline{g}}}{m_{\overline{g}}^{'}}$ is at least $it-3pi$. (We prove this fact in Lemma~\ref{lem:distderivative}). 
To prove the lower bound, we will show that there is a ``large'' nice subset ${\cal N} \subseteq {{\cal S}_p}^{\tilde{t}}$ such any $\overline{f}$ and $\overline{g}$ in ${\cal N}$ differ in a constant fraction of all coordinates. The following lemma, which just follows from the existence and properties of Reed-Solomon codes guarantees the existence of such an $\cal N$. 

\begin{lem}~\label{lem:reedsolomon}
Let $0 < \alpha < 1$ be any absolute constant and let $q$ be a prime power. For any alphabet $\Sigma$ of size $q$ and positive integer $m$ such that $m < q$, there is a set $\cal C$ of strings of length $m$ over $\Sigma$ of size $q^{{(1-\alpha)m}}$ such that any two strings in $\cal C$ are at a distance at least $\alpha m$ apart. 
\end{lem}
\begin{proof}
Let $\cal C$ be the set of codewords obtained when the set $\Sigma^{(1-\alpha)m}$ is encoded using Reed-Solomon codes of message length $(1-\alpha){m}$ and code length $m$. The distance of the code is $\alpha m$ and the number of codewords is $q^{{(1-\alpha)m}}$. Hence the set satisfies the properties stated in the statement.  
\end{proof}

Lemma~\ref{lem:reedsolomon} immediately implies the existence of a set $\cal N$, when invoked with parameters $\Sigma = {{\cal S}_p}$, $m = \tilde{t}$. So, we have the following corollary.

\begin{cor}~\label{lem:niceset}
For all $\alpha$ such that $0< \alpha < 1$, there exists ${\cal N} \subseteq {{\cal S}_p}^{\tilde{t}}$ of size equal to $n^{{(1-\alpha)(p+1)\tilde{t}}} $ such that for any distinct pair $\overline{f}$ and $\overline{g}$  in ${\cal N}$, $\overline{f}$ and $\overline{g}$ differ in at least $\alpha{\tilde{t}}$ coordinates.
\end{cor}

Informally, the set $\cal N$ now gives us a large number of partial derivatives which are at a large distance from each other. We formalize this claim in the lemma below. 

\begin{lem}~\label{lem:distderivative}
For $k = 2p\tilde{t}$, the set $\partial^{=k} {\cal P}_{p,t,n}$ has a subset $S$ of size at least $n^{(1-\alpha){(p+1)\tilde{t}}}$ such that every element in this subset is a monomial and any two such monomials are at a distance of at least $\alpha{\tilde{t}}(t-3p)$ from each other. 
\end{lem}
\begin{proof}
Let us pick any two $\overline{f}$ and $\overline{g}$ in $\cal N$. Let $i \in [\tilde{t}]$ be an index such that $f_i \neq g_i$. Then over the set $C_i$, $f_i$ and $g_i$ can agree at at most $p$ points. Therefore, the monomials $m_{\overline{f}}$ and $m_{\overline{g}}$ differ in at least $t-p$ variables of the form $x_{h,j}$ for $h \in C_i$. Now, for each $i$ and each $\overline{f} \in {{\cal S}_p}^{\tilde{t}}$, $m_{\overline{f}}^{'}$ contains exactly  $2p$ variables $x_{h,j}$ with $h \in C_i$. Hence, $\frac{m_{\overline{f}}}{m_{\overline{f}}^{'}}$ and $\frac{m_{\overline{g}}}{m_{\overline{g}}^{'}}$ differ in at least $t-3p$ variables of the form $x_{h,j}$ for $h \in C_i$. So, each coordinate $i$ where $\overline{f}$ and $\overline{g}$ differ from each other contributes $t-3p$ to the distance between $\frac{m_{\overline{f}}}{m_{\overline{f}}^{'}}$ and $\frac{m_{\overline{g}}}{m_{\overline{g}}^{'}}$. Hence, for every $\overline{f}$ and $\overline{g}$ 
$\in {\cal N}$,  $\frac{m_{\overline{f}}}{m_{\overline{f}}^{'}}$ and $\frac{m_{\overline{g}}}{m_{\overline{g}}^{'}}$ are at a distance at least $\alpha{\tilde{t}}(t-3p)$ apart. The lemma now follows from the fact that the size of $\cal N$ is at least $n^{(1-\alpha){(p+1)\tilde{t}}}$. 
\end{proof}

We now essentially have all the ingredients we need for showing lower bounds for homogeneous $\hsp$ circuits computing ${\cal P}_{p,t,n}$.  We will use the following lemma which is implicit in~\cite{CM13}. A similar calculation also appears in~\cite{FLMS13}.

\begin{lem}[Implicit in~\cite{CM13}]~\label{lem:lb}
Let $Q = \sum_{i = 1}^{s^{'}} {Q_{i1}Q_{i2}\ldots Q_{iz}}$ where each $Q_{ij}$ is an $N$ variate polynomial of degree bounded by $u$. Also, for some $r \leq z$, suppose there are $s$ elements in $\partial^{=r} Q$ such that the distance between the leading monomial of any two of these is at least $d$. Then, for any positive integer $\ell$ such that $\ell \leq \frac{Nd}{2\ln(s.N^2)}$ and $(ru-r)^2 = o(\ell)$, 

$$s^{'} \geq \frac{s(1-\frac{1}{N^2})}{{z+r \choose r}{e^{\frac{N(ru-r)}{\ell}}}}$$. 

\end{lem}

For $p = 1$, we will call the corresponding polynomial family $\{{\cal P}_{p,t,n}\}_n$ as $\{{\cal P}_{t,n}\}_n$. We will now prove the following lower bound for homogeneous $\hspn$ circuits computing ${\cal P}_{t,n}$.
\begin{thm}~\label{thm:mainthm0}
For any $ \omega(\log n) \leq t \leq \frac{n}{40}$, any homogeneous $\hspn$ circuit computing ${\cal P}_{t,n}$ has top fan-in at least $2^{\Omega{(\frac{n}{t}\log n})}$. 
\end{thm}

\begin{proof}
We will invoke Lemma~\ref{lem:lb} after setting the parameters used in it appropriately. Let
\begin{itemize}
\item $p = 1$
\item $\alpha = 0.9$
\item $N = n^2$
\item $u = \frac{t}{20}$
\item $z = O(\frac{n}{t})$
\item $r = \frac{2n}{t}$
\end{itemize}
Now, Lemma~\ref{lem:distderivative} implies that there is a set $S$ of size $s = n^{0.2\tilde{t}}$ such that any two monomials in $S$ are at a distance at least $d = \alpha{\tilde{t}}(t-3p) = 0.9n - 2.7\frac{n}{t}$. Observe that for $t > 270$, we get $d \geq 0.89n$. We now need to set $\ell$ to a value which satisfies the constraints in the hypothesis of Lemma~\ref{lem:lb} and which also implies a non trivial lower bound on $s^{'}$. The hypothesis of Lemma~\ref{lem:lb} requires that $\ell \leq \frac{Nd}{2\ln(s.N^2)}$. Substituting the values of $p, N, s, d$, we require $\ell \leq \frac{n^2\times 0.89 n}{2(4 + \frac{0.2n}{t})\ln n}$. For $t < \frac{n}{20}$, any $\ell < \frac{0.89n^2t}{0.8\ln t}$ will satisfy this constraint since, $\frac{n^2\times 0.89 n}{2(4 + \frac{0.2n}{t})\ln n} \geq \frac{0.89n^2t}{0.8\ln t}$. Observe that the term ${z+r \choose r}$ is of the order $2^{O{(\frac{n}{t})}}$ by Shannon's entropy estimation. The other high order term in the denominator is $e^{\frac{Nru}{\ell}} = e^{\frac{n^3}{10\ell}}$. On the other hand, the highest order term in the numerator is $s = n^{0.2\frac{n}{t}}$. So, for a lower bound of the order  $n^{\Omega(\frac{n}{t})}$, we will ensure that $0.2\frac{n\ln n}{t} \geq 1.1{\frac{n^3}{10\ell}}$. This requires $\ell > \frac{0.55n^2t}{\ln n}$. So, we need $\frac{0.55n^2t}{\ln n} <\ell< \frac{0.89n^2t}{0.80\ln t}$ and $o(\ell) = n^2$. Let us set $\ell = \frac{n^2t}{\ln n}$. For $t$ being $\omega(\log n)$, this satisfies $o(\ell) = n^2$. Substituting all these values into the expression in Lemma~\ref{lem:lb}, we get $s^{'} \geq 2^{\Omega(\frac{n}{t}\log n)}$. 
\end{proof}
This completes the proof of Theorem~\ref{thm:hierarchy}. We will now build upon this proof to obtain Theorem~\ref{thm:dr}.
\subsection{Proof of Theorem~\ref{thm:dr}}
So far, we have constructed a polynomial family ${\cal P}_{t,n}$ such that ${\cal P}_{t,n}$ requires homogeneous $\hspn$ circuits with top fan-in at least $n^{\Omega(\frac{n}{t})}$. We can now build upon the construction of ${\cal P}_{t,n}$ described so far to construct a single polynomial family which is hard for any homogeneous $\hsp$ circuit for every $\omega(\log n) \leq a \leq \frac{n}{800}$. We will now use a variation of the interpolation trick described in Lemma 14 in~\cite{KSS13}. The idea now is just to take a linear combination of ${\cal P}_{a,n}$ for $O(\log n)$ many such values of $a$, with coefficients being the variables $\overline{y}$, such that for every $a$ such that $\omega(\log n) \leq a \leq \frac{n}{800}$, there is a $t$ such that $ 20a\leq t \leq 400a$ and such that ${\cal P}_{t,n}$ is in the linear combination. 

In particular let us define the following family of polynomials ${\cal Q}_{n}$:

$$ {\cal Q}_{n}(\overline{x}, \overline{y}) = \sum_{i = 0}^{O(\log n)} y_i\cdot{\cal P}_{{20}^i,n}(\overline{x})$$
Observe that ${\cal Q}_{n}$ can be computed by a polynomial size homogeneous $\spsp{(\log n)}$ circuit.

If  ${\cal Q}_{n}$ could be computed {\it efficiently} by a homogeneous $\spsp^{[a]}$ for some $a$, then so could any projection of the sum (i.e. we set all but one of the $y_i$ to $0$), i.e. so could ${\cal P}_{t,n}$. This contradicts Theorem~\ref{thm:hierarchy}. In particular we get that  every $\omega(\log n) \leq a \leq \frac{n}{800} $, any homogeneous $\hsp$ circuit computing ${\cal Q}_{n}$ must have top fan-in at least $2^{\Omega(\frac{n}{a}\log n)}$. 
This completes the proof of Theorem~\ref{thm:dr}.

\subsection{Proof of Theorem~\ref{thm:regular}}

The proof of Theorem~\ref{thm:regular} follows immediately from Theorem~\ref{thm:dr} along with Theorem 15 from~\cite{KSS13}.

\section{Discussion and future directions}~\label{sec:conclusion}
One of the main questions left open by our lower bounds on $\spsp(r)$ circuits is to remove the restriction on top fan-in, and to prove super polynomial lower bounds for all homogeneous $\spsp$ circuits. Currently, we have no nontrivial lower bounds for homogeneous  $\spsp(\log n)$ circuits, even in the further special case when the family of circuits is a sum of $\psp^{[t]}$ circuits for different values of $t$. We identify this as the simplest class of circuits/formulas for which we don't know how to prove lower bounds. While this would still not suffice in proving lower bounds for general arithmetic circuits, this seems to be an important step in that direction.  Another very interesting direction would be to give nontrivial PIT results for $\spsp(r)$ circuits when $r$ is a constant. So far, we only know how to derandomize PIT when the $\spsp(r)$ circuits are multilinear, and our lower bound for $\spsp(r)$ circuits could be viewed as a first step in this direction. 



One corollary of our results is a hierarchy theorem for $\spsp^{[t]}$ formulas. A very interesting question that we don't know how to answer is if there is a tighter hierarchy theorem. We believe that for every $t$, polynomial sized $\spsp^{[t]}$ formulas should be able to compute a much richer class of polynomials than polynomial sized $\spsp^{[t-1]}$ formulas. A special case that we do not know how to answer is the relative complexity of $\spsp^{[2]}$ formulas versus $\spsp^{[1]}$ formulas (which are basically depth 3 formulas). Another kind of hierarchy question that we don't fully understand but which we think would be very interesting is to understand the relative complexity of depth $d$ formulas versus depth $d+1$ formulas for constant $d$. Perhaps a refinement of the depth reduction techniques of Koiran and Tavenas would shed light on these questions. 

\section*{Acknowledgements}
We would like to thank Klim Efremenko, Amir Shpilka and Amir Yehudayoff  for helpful discussions at several stages of this work. We would also like to thank Swastik Kopparty and Avi Wigderson for many helpful comments on an earlier version of this paper.

\bibliography{refs}

\appendix

\section{Completeness of the model of $\spsp(r)$ circuits}\label{app:completeness}

Depth 3 and depth 4 circuits with bounded top fan-in ($\sps(r)$ and $\spsp(r)$ respectively) have been extensively studied in the past especially in the context of polynomial identity testing (PIT). The question of lower bound for $\sps(r)$ circuits is almost uninteresting since it can be shown quite easily that for $r < n$, $\sps(r)$ circuits cannot compute the $n \times n$ permanent or determinant, {\it no matter what the size} of the circuit. Thus the class of $\sps(r)$ circuits is not complete, in the sense that
the class of circuits cannot even compute all polynomials. In contrast, the class of depth $2$ $\sp$ circuits (with no restriction on top fan-in) is complete, but lower bounds are trivial for this model since any polynomial with $m$ monomials needs a $\sp$ circuit of size at least $m$ to compute it. 

It was observed by Kayal~\cite{kayal} that if one considers the class of depth 4 circuits and one imposes the additional requirement that each product gate has at least $2$ nontrivial factors, then the class of $\spsp(r)$ circuits with $r< n/2$ circuits is not complete. This is because if $\alpha$ is a common root of at least two of the factors of each of the product gates, then it would be a zero of multiplicity $2$ of the polynomial computed by the circuit. Also if $r < n/2$ then such an $\alpha$ always exists. Hence if one starts with a polynomial that does not vanish at any point with multiplicity $2$, then it cannot be computed by such a circuit. Thus in this case we can prove lower bounds easily. 

The general class of depth $4$ $\spsp(r)$ circuits even when $r=1$ is a complete class, since it  contains the class of depth 2 $\sp$ circuits.  However lower bounds for $\spsp(r)$ circuits for $r \geq 2$  did not seem to be known prior to this work.

\end{document}